\documentclass[lettersize,journal]{IEEEtran}
\IEEEoverridecommandlockouts
\usepackage{cite}
\usepackage{amsmath,amssymb,amsfonts}
\usepackage{graphicx}
\usepackage{textcomp}
\usepackage{booktabs}
\usepackage{array}
\usepackage{amsthm}
\usepackage[noend, ruled, linesnumbered]{algorithm2e} 
\usepackage{bm}
\usepackage{booktabs}
\usepackage{inconsolata}
\usepackage{setspace}
\usepackage{float}
\usepackage[colorlinks,bookmarksopen,bookmarksnumbered,citecolor=blue, linkcolor=blue, urlcolor=blue]{hyperref} 
\usepackage[caption=false,font=footnotesize]{subfig}
\usepackage{multirow}           
\usepackage{threeparttable}     
\usepackage{enumerate}
\usepackage{makecell}
\usepackage{setspace}
\usepackage{balance}
\usepackage{comment}
\usepackage{placeins}
\usepackage{makecell}
\usepackage{xcolor}

\SetKwFor{For}{For}{do}{endfor}
\SetKwIF{If}{ElseIf}{Else}{If}{then}{else if}{Else}{endif}
\SetKwFor{While}{While}{do}{endwhile}
\SetKw{Return}{Return}

\def\BibTeX{{\rm B\kern-.05em{\sc i\kern-.025em b}\kern-.08em
		T\kern-.1667em\lower.7ex\hbox{E}\kern-.125emX}}

\newtheoremstyle{noparens} 
{2pt} 
{2pt} 
{\normalfont} 
{1em} 
{\bfseries\itshape} 
{.} 
{ } 
{\thmname{#1} \thmnumber{#2}\mdseries\thmnote{ (\hspace{-0.25pt}#3)}} 
\theoremstyle{noparens}

\newtheorem{theorem}{{\textbf{Theorem}}}

\makeatletter

\newcommand{\Rmnum}[1]{\expandafter\@slowromancap\romannumeral #1@}
\makeatother

\newtheoremstyle{roman} 
{2pt} 
{2pt} 
{\normalfont} 
{1em} 
{\bfseries\itshape} 
{.} 
{ } 
{\thmname{#1} \Rmnum{#2}\mdseries\thmnote{ (\hspace{-0.25pt}#3)}} 
\theoremstyle{roman}

\makeatletter
\renewenvironment{proof}[1][\proofname]{\par
	\pushQED{\qed}%
	\normalfont \topsep2\p@\@plus2\p@\relax
	\trivlist
	\item[\hskip\labelsep
	\hskip\dimexpr1.5em\relax 
	\itshape
	#1\@addpunct{:}]\ignorespaces
}{%
	\popQED\endtrivlist\@endpefalse
}
\makeatother

\begin{document}
\title{SCL Decoding of Non-Binary Linear Block Codes}
\author{Jingyu Lin, Li Chen, \textit{Senior Member, IEEE}, and Xiaoqian Ye
\thanks{
	
Jingyu Lin and Li Chen are with the School of Electronics and Information Technology, Sun Yat-sen University, Guangzhou 510006, China and Guangdong Province Key Laboratory of Information Security Technology, Guangzhou 510006, China (e-mail: linjy228@mail2.sysu.edu.cn; chenli55@mail.sysu.edu.cn). Xiaoqian Ye is with the School of Electronics and Information Technology, Sun Yat-sen University, Guangzhou 510006, China (e-mail: yexq26@mail2.sysu.edu.cn).}}

\markboth{IEEE Communications Letters, submitted}%
{How to Use the IEEEtran \LaTeX \ Templates}

\maketitle

\begin{abstract}
 Non-binary linear block codes (NB-LBCs) are an important class of error-correcting codes that are especially competent in correcting burst errors. They have broad applications in modern communications and storage systems. However, efficient soft-decision decoding of these codes remains to be further developed. This paper proposes successive cancellation list (SCL) decoding for NB-LBCs that are defined over a finite field of characteristic two, i.e., $\mathbb{F}_{2^r}$, where $r$ is the extension degree. By establishing a one-to-$r$ mapping between the binary composition of each non-binary codeword and $r$ binary polar codewords, SCL decoding of the $r$ polar codes can be performed with a complexity that is sub-quadratic in the codeword length. A simplified path sorting is further proposed to facilitate the decoding. Simulation results on short-length extended Reed-Solomon (eRS) and non-binary extended BCH (NB-eBCH) codes show that SCL decoding can outperform their state-of-the-art soft-decision decoding with fewer finite field arithmetic operations. For length-$16$ eRS codes, their maximum-likelihood (ML) decoding performances can be approached with a moderate list size.
\end{abstract}

\begin{IEEEkeywords}
Non-binary linear block codes, successive cancellation list decoding, soft-decision decoding.
\end{IEEEkeywords}

\section{Introduction}
\IEEEPARstart{N}{on-binary} linear block codes (NB-LBCs) are an important class of error-correcting codes with wide applications in modern communications and storage systems. They are especially competent in correcting burst errors. The celebrated NB-LBCs include Reed-Solomon (RS) codes\cite{reed1960polynomial}, algebraic-geometry (AG) codes\cite{goppa1977codes}, and non-binary BCH (NB-BCH) codes\cite{bose1960class,hocquenghem1959codes}. Their algebraic decoding can be categorized into the syndrome-based approach and the curve-fitting-based approach. The former can efficiently correct errors up to half the code's minimum Hamming distance, leading to the Berlekamp-Massey (BM) algorithm \cite{berlekamp, massey2003shift} being widely adopted in RS coded systems. The latter is also known as the Guruswami-Sudan (GS) algorithm\cite{GS}. It improves the error-correction capability beyond this bound while retaining a polynomial-time decoding complexity. 

Further performance improvement can be achieved by utilizing soft information observed from the channel. The algebraic soft-decision decoding (ASD), also known as the K\"otter-Vardy (KV) decoding \cite{koetter2003algebraic}, enhances the decoding of RS codes by converting the reliability information into the interpolation multiplicities. For RS codes of length $N$, KV decoding with a maximum output list size $L$ exhibits a complexity of $\mathcal{O}(N^2L^5)$\cite{xingchen}.  Another soft-decision decoding approach is the Chase decoding\cite{Chase}. It constructs  $2^{\eta}$ test-vectors by flipping the $\eta$ least reliable symbols. With each test-vector decoded by the BM algorithm, the Chase-BM decoding exhibits a complexity of $\mathcal{O}(2^\eta N^2)$. Hence, soft-decision decoding of NB-LBCs are generally more complex. Efficient soft-decision decoding approach for NB-LBCs remains to be further  developed.

In \cite{dynamicfrozen}, it was revealed that any linear block code can be interpreted as a polar code with dynamic frozen symbols. Subsequently, successive cancellation (SC) decoding of RS codes was proposed in \cite{RS_SC}, where RS codes are transformed into non-binary polar codes. Sequential decoding and permutation-based decoding were applied in \cite{RS_SC, RS_SCS, RS_SCP} to further improve the decoding performance. Recent research proposed a general transformation from binary linear block codes (B-LBCs) to polar codes with dynamic frozen symbols\cite{linhuang, linhuang2}. It established a one-to-one mapping between B-LBC codewords and polar codewords, in which a permutation matrix is required to adjust the information set of the polar code. Consequently, SC and SC list (SCL) \cite{tal2015list, niu2012crc} decoding, and their fast decoding version\cite{Simplified, broadcast, fastSCL_journal}, can be applied to decode B-LBCs. 

In this paper, SCL decoding is proposed for NB-LBCs that are defined over a finite field of characteristic two, i.e., $\mathbb{F}_{2^r}$, where $r$ is the extension degree. By establishing a one-to-$r$ mapping between the binary composition of each non-binary codeword and $r$ binary polar codewords, SCL decoding of the $r$ polar codes can be performed with a complexity of $\mathcal{O}(rLN\text{log}_2N)$, where $L$ is the SCL decoding list size. A simplified  path sorting is further proposed to facilitate the decoding. Simulation results on short-length extended RS (eRS) codes and non-binary extended BCH (NB-eBCH) codes show that SCL decoding can outperform their state-of-the-art soft-decision decoding with fewer finite field arithmetic operations. Moreover, for length-16 eRS codes, their maximum-likelihood (ML) decoding performances can be approached with a moderate list size.

\textbf{Notation:} Let $\mathbb{F}_2$ denote the binary field and $\mathbb{F}_{2^r}$ subsequently denote its extension field of extension degree $r$. Further let $p(X)$ and $\alpha$ denote the primitive polynomial and the primitive element of $\mathbb{F}_{2^r}$, respectively. Given an element $\sigma \in \mathbb{F}_{2^r}$, it can be represented by $\sum_{j=0}^{r-1} \sigma_j \alpha^j$, where $\sigma_j \in \mathbb{F}_2$. Vector $(\sigma_0, \sigma_1, \cdots, \sigma_{r-1})$ is the binary composition of $\sigma$. For convenience, we also use $\sigma[j]$ to denote  $\sigma_{j}$. Given an integer set $\mathcal{A} \subset \{0,1,\cdots,N-1\}$, its cardinality and complement are denoted by $|\mathcal{A}|$ and $\mathcal{A}^c$, respectively.

\section{Preliminaries}

\subsection{Polar Codes}
Let us consider a polar code of length $N = 2^n$ and dimension $K$. With kernel matrix $\bm{{\rm F}} = ((1, 0), (1, 1))^{\rm T}$, its generator matrix is $\bm{{\rm G}}_{\rm p} = \bm{{\rm F}}^{\otimes n}$, where $\otimes n$ denotes the $n$-fold Kronecker product\cite{arikan2009polar}. With an input vector $\bm{u} = (u_0, u_1, \cdots, u_{N-1}) \in \mathbb{F}_{2}^N$, which is constituted by $K$ information symbols and $N-K$ frozen symbols, its codeword $\bm{c}$ is generated by
\begin{equation}
	\bm{c} = \bm{u}\bm{{\rm G}}_{\rm p}.
\end{equation}
The indices of information and frozen symbols constitute the information set $\mathcal{A}$ and the frozen set $\mathcal{A}^c$, respectively. Let $\bm{u}^{\mathcal{A}} = (u_i | i\in \mathcal{A})$ denote the message. Codeword $\bm{c}$ can be alternatively represented by
\begin{equation} \label{DF}
	\bm{c} = \bm{u}^{\mathcal{A}}\bm{{\rm M}}\bm{{\rm G}}_{\rm p},
\end{equation}
where $\bm{{\rm M}} \in \mathbb{F}_{2}^{K\times N}$ is a pre-transformed matrix of reduced row echelon form. Note that indices of the pivot columns in $\bm{{\rm M}}$ form $\mathcal{A}$. The frozen symbol with index $i$ is a linear combination of information symbols with indices smaller than $i$, i.e.,
\begin{equation} \label{df}
	u_i = \sum_{t=0}^{\tau_i} u^{\mathcal{A}}_t \cdot \bm{{\rm M}}_{t,i}, \, i\in \mathcal{A}^c,
\end{equation}
where $\tau_i = |\mathcal{A}\cap \{0,\cdots,i-1\}|$. It is referred to as a dynamic frozen symbol \cite{dynamicfrozen}. During SC and SCL decoding, all $u^{\mathcal{A}}_t$ with $t\leq \tau_i$ are estimated prior to $u_i$, which allows $u_i$ to be determined accordingly.

\subsection{Transformation From B-LBCs To Polar Codes}
For a B-LBC $\mathcal{C}_{\rm B}$ of length $N=2^n$ and dimension $K$, there exists a polar code with dynamic frozen symbols such that the one-to-one mapping between codewords of $\mathcal{C}_{\rm B}$ and the polar codewords can be established \cite{linhuang}. Let $\bm{{\rm G}}_{\rm B} \in \mathbb{F}_{2}^{K \times N}$ denote a generator matrix of $\mathcal{C}_{\rm B}$ and $\bm{m} \in \mathbb{F}_{2}^K$ denote a $K$-dimensional message. Given a permutation matrix $\bm{{\rm P}}\in \mathbb{F}_{2}^{N \times N}$, codebook $\mathcal{C}_{\rm B}$ can be defined as
\begin{equation} \label{BLBC}
	\begin{aligned}
		\mathcal{C}_{\rm B} &\triangleq \{\bm{c} = \bm{m}\bm{{\rm G}}_{\rm B} \: | \: \forall \bm{m} \in \mathbb{F}_{2}^K\} \\
		&= \{\bm{c} = \bm{u}\bm{{\rm G}}_{\rm p}\bm{{\rm P}} \: | \: \bm{u} = \bm{m}\bm{{\rm G}}_{\rm B}\bm{{\rm P}}^{-1}\bm{{\rm G}}_{\rm p}^{-1}, \forall \bm{m} \in \mathbb{F}_{2}^K\}.	
	\end{aligned}
\end{equation}
By performing Gaussian elimination (GE) on $\bm{{\rm G}}_{\rm B}\bm{{\rm P}}^{-1}\bm{{\rm G}}_{\rm p}^{-1}$, the pre-transformed matrix $\bm{{\rm M}} = \bm{{\rm E}}\bm{{\rm G}}_{\rm B}\bm{{\rm P}}^{-1}\bm{{\rm G}}_{\rm p}^{-1}$ can be obtained, where $\bm{{\rm E}} \in \mathbb{F}_{2}^{K \times K}$ is a row elimination matrix. Hence, 
\begin{equation} 
	\mathcal{C}_{\rm B} \triangleq \{\bm{c} = \bm{u}^{\mathcal{A}}\bm{{\rm M}}\bm{{\rm G}}_{\rm p}\bm{{\rm P}} \: | \: \forall \bm{u}^{\mathcal{A}} \in \mathbb{F}_{2}^K\}, 
\end{equation}
where $\bm{u}^{\mathcal{A}} = \bm{m}\bm{{\rm E}}^{-1}$. Therefore, each codeword of $\mathcal{C}_{\rm B}$ can be mapped to a permuted polar codeword with dynamic frozen symbols. $\mathcal{C}_{\rm B}$ can be decoded by estimating the polar codewords through SC or SCL decoding.

\section{Decomposition of NB-LBCs}
To facilitate SCL decoding of NB-LBCs, a mapping between NB-LBCs and binary polar codes is required. It is established by the following Theorem \ref{mapping}.
\begin{theorem} \label{mapping}
	Let $\mathcal{C}_{\rm NB}$ denote an $(N=2^n, K)$ NB-LBC defined over $\mathbb{F}_{2^r}$. The binary composition of any  codeword $\bm{c} \in \mathcal{C}_{\rm NB}$ can be represented as the concatenation of $r$ permuted binary polar codewords.
\end{theorem}
\begin{proof}
	Let $\bm{{\rm G}}_{\rm NB} \in \mathbb{F}_{2^r}^{K \times N}$ denote a generator matrix of $\mathcal{C}_{\rm NB}$. With a message $\bm{m} \in \mathbb{F}_{2^r}^K$, its codeword $\bm{c}$ is generated by $\bm{c} = \bm{m}\bm{{\rm G}}_{\rm NB}$. Given a permutation matrix $\bm{{\rm P}}$, $\bm{c}$ can also be represented by 
	\begin{equation} \label{cug}
		\bm{c} = \bm{u}\bm{{\rm G}}_{\rm p}\bm{{\rm P}},
	\end{equation}
	where $\bm{u} = \bm{m}\bm{{\rm G}}_{\rm NB}\bm{{\rm P}}^{-1}\bm{{\rm G}}_{\rm p}^{-1}$. Let 
	\begin{equation}
		\begin{aligned}
			\bm{u}^{\rm B} = (u_{0,0},\, \cdots,\, u_{0,r-1},\, &u_{1,0},\, \cdots,\, u_{1,r-1},\\ 
			\cdots,\,&u_{N-1,0},\, \cdots,\, u_{N-1,r-1})
		\end{aligned}
	\end{equation}
	denote the binary composition of $\bm{u}$. Furthermore, let
	\begin{equation}
		\begin{aligned}
			\bm{c}^{\rm B} = (c_{0,0},\, \cdots,\, c_{0,r-1},\, &c_{1,0},\, \cdots,\, c_{1,r-1},\\ 
			\cdots,\,&c_{N-1,0},\, \cdots,\, c_{N-1,r-1})
		\end{aligned}
	\end{equation}
	denote the binary composition of $\bm{c}$. Since $\bm{{\rm G}}_{\rm p}$ and $\bm{{\rm P}}$ are binary matrices, the multiplication between $\bm{u}$ and $\bm{{\rm G}}_{\rm p}\bm{{\rm P}}$ involves only $\mathbb{F}_{2^r}$ additions, which can be decomposed into $\mathbb{F}_{2}$ additions. For $j = 0, 1, \cdots, r-1$, let
	\begin{align}
		\bm{u}^{\rm B}_j &= (u_{0,j},\, u_{1,j},\,\cdots,\, u_{N-1,j}), \\
		\bm{c}^{\rm B}_j &= (c_{0,j},\, c_{1,j},\,\cdots,\, c_{N-1,j}).
	\end{align}
	Based on (\ref{cug}), one can obtain
	\begin{equation}
		\bm{c}^{\rm B}_j = \bm{u}^{\rm B}_j \bm{{\rm G}}_{\rm p}\bm{{\rm P}}.
	\end{equation}
	By performing GE on $\bm{{\rm G}}_{\rm NB}\bm{{\rm P}}^{-1}\bm{{\rm G}}_{{\rm p}}^{-1}$, the non-binary pre-transformed matrix
	$\bm{{\rm T}} = \bm{{\rm R}}\bm{{\rm G}}_{\rm NB}\bm{{\rm P}}^{-1}\bm{{\rm G}}_{{\rm p}}^{-1}$ can be obtained, where $\bm{{\rm R}} \in \mathbb{F}_{2^r}^{K \times K}$ is a row elimination matrix. Consequently, 
	\begin{equation}
		\bm{u} = \bm{u}^{\mathcal{A}}\bm{{\rm T}},
	\end{equation}
	where $\bm{u}^{\mathcal{A}} = \bm{m}\bm{{\rm R}}^{-1}$. Hence, $\bm{c}^{\rm B}_0, \bm{c}^{\rm B}_1, \cdots, \bm{c}^{\rm B}_{r-1}$ share an identical information set, i.e., indices of the pivot columns in $\bm{{\rm T}}$. The frozen symbols are determined as in (\ref{df}) but through $\mathbb{F}_{2^r}$ additions and multiplications. Therefore, $\bm{c}^{\rm B}$ is the concatenation of $r$ permuted binary polar codewords.
\end{proof}
Therefore, $\mathcal{C}_{\rm NB}$ can be decoded by estimating the $r$ binary polar codewords through SC or SCL decoding and then reconstructing the non-binary codeword. Obtaining the pre-transformed matrix $\bm{{\rm T}}$ requires a complexity of $\mathcal{O}(KN\text{log}_2N + NK^2)$. Since this can be performed  offline, it does not affect the overall decoding complexity. 

Theoretically, the SC decoding error probability of binary polar codes is upper bounded by $P^{\rm UB}_{\rm e} = \sum_{i\in\mathcal{A}}P_{\rm e}(\mathcal{W}_i)$, where $\mathcal{W}_i$ denotes the $i$-th polarized subchannel of a length-$N$ polar code and $P_{\rm e}(\mathcal{W}_i)$ denotes its error probability\cite{arikan2009polar}. This implies that the SC decoding performance of $\mathcal{C}_{\rm NB}$ is determined by the information set $\mathcal{A}$, which is further determined by the permutation matrix $\bm{{\rm P}}$. Hence, a permutation matrix that minimizes $P^{\rm UB}_{\rm e}$ should be chosen to optimize the SC decoding performance. Unfortunately, optimal design of $\bm{{\rm P}}$ remains to be proven. According to \cite{RS_SC, linhuang}, there exists an effective permutation for eRS codes, which is defined as
\begin{equation} \label{permutation}
	\bm{{\rm P}}_{a, b} = 
	\begin{cases}
		1, &\text{if} \:a = \sum_{j=0}^{m-1}(\alpha^b)[j]\cdot 2^{j}, 0 \leq b \leq N-2\\ &
		\:\text{or}\:a = 0, b = N-1; \\
		0, &\text{otherwise},
	\end{cases}
\end{equation}
where $\bm{{\rm P}}_{a,b}$ is the row-$a$ column-$b$ entry of $\bm{{\rm P}}$. It results in the first few subchannels being frozen, which helps reduce $P^{\rm UB}_{\rm e}$ since their error probabilities are relatively high. E.g., for code rate between  $0.25$ and $0.5$, $\mathcal{W}_0$, $\mathcal{W}_1$ and $\mathcal{W}_2$ are frozen regardless of code length, implying that the first information symbol is transferred over $\mathcal{W}_3$. Due to channel polarization \cite{arikan2009polar}, $P_{\rm e}(\mathcal{W}_3)$ increases with $N$, resulting in a corresponding increase in $P^{\rm UB}_{\rm e}$. Hence, when $\bm{{\rm P}}$ of (\ref{permutation}) is applied, SC decoding performance of eRS codes degrades as $N$ increases. In this paper, this $\bm{{\rm P}}$ is applied in SCL decoding of short-length eRS and NB-eBCH codes, which are extended to length $2^n$ by padding a parity symbol as $c_{2^n-1} = \sum_{i=0}^{2^n-2} c_i$. 

\section{SCL Decoding of NB-LBCs}

\subsection{SC and SCL Decoding}

Assume that codeword $\bm{c} = (c_0, c_1, \cdots, c_{N-1}) \in \mathbb{F}_{2^r}^N$ of $\mathcal{C}_{\rm NB}$ is transmitted over a memoryless channel and $\bm{y} = (y_0, y_1, \cdots, y_{N-1}) \in \mathbb{R}^N$ is the received vector. Let $\bm{c}^{\rm P} = \bm{c}\bm{{\rm P}}^{-1}$ and $\bm{y}^{\rm P} = \bm{y}\bm{{\rm P}}^{-1}$. Let
\begin{equation}
	\begin{aligned}
		\bm{\mathcal{L}} = (\mathcal{L}_{0,0},\, \cdots,\, \mathcal{L}_{0,r-1},\, &\mathcal{L}_{1,0},\, \cdots,\, \mathcal{L}_{1,r-1},\\ 
		\cdots,\, &\mathcal{L}_{N-1,0},\, \cdots,\, \mathcal{L}_{N-1,r-1})
	\end{aligned}
\end{equation}
denote the log-likelihood ratio (LLR) vector with entries defined as
\begin{equation}
	\mathcal{L}_{i,j} = \text{ln}\frac{p(y_i^{\rm P}|c_{i,j}^{\rm P}=0)}{p(y_i^{\rm P}|c_{i,j}^{\rm P}=1)},
\end{equation}
where $i = 0, 1, \cdots, N-1$ and $j = 0, 1, \cdots, r-1$. These LLRs are partitioned into $r$ groups, each of which is the input LLR vector of an SC decoder. In particular, the input LLR vector of the $j$-th SC decoder is
\begin{equation}
	\bm{\mathcal{L}}_{j}^{(n)} = (\mathcal{L}_{j,0}^{(n)},\, \mathcal{L}_{j,1}^{(n)},\, \cdots,\, \mathcal{L}_{j,N-1}^{(n)}),
\end{equation}
where $\mathcal{L}_{j,i}^{(n)} = \mathcal{L}_{i,j}$. For $0\leq s \leq n-1$, the stage-$s$ LLRs are computed by \cite{llr_scl}
\begin{equation} \label{llr_update}
	\begin{aligned}
		\mathcal{L}_{j,i}^{(s)} &= f(\mathcal{L}_{j,i}^{(s+1)}, \mathcal{L}_{j,i+2^s}^{(s+1)}), \\
		\mathcal{L}_{j,i+2^s}^{(s)} &= (-1)^{\hat{u}_{j,i}^{(s)}}\mathcal{L}_{j,i}^{(s+1)} + \mathcal{L}_{j,i+2^s}^{(s+1)},
	\end{aligned}
\end{equation}
where $f(\mathcal{X}, \mathcal{Y}) \triangleq \text{ln}\frac{e^\mathcal{X}e^\mathcal{Y}+1}{e^\mathcal{X}+e^\mathcal{Y}}$ and $\mathcal{X},\mathcal{Y}\in \mathbb{R}$. When reaching stage-$0$, hard decisions are made based on the LLRs, i.e.,
\begin{equation}
	\eta_{j, i} = 
	\begin{cases}
		0, \quad &\text{if}\, \mathcal{L}_{j, i}^{(0)}\geq 0; \\
		1, \quad &\text{otherwise}.
	\end{cases}
\end{equation}
Then the binary input symbols are estimated by 
\begin{equation} 
	\hat{u}_{j,i}^{(0)} = 
	\begin{cases}
		\eta_{j, i},\quad &\text{if}\, i \in \mathcal{A}; \\
		(\sum_{t=0}^{\tau_i} \hat{u}^{\mathcal{A}}_t \cdot \bm{{\rm T}}_{t,i})[j],\quad &\text{if}\, i \in \mathcal{A}^c,
	\end{cases}
\end{equation}
where $\hat{u}^{\mathcal{A}}_t$ is the estimation of the $t$-th non-binary information symbol. The estimations of the non-binary input symbols is determined by
\begin{equation}
	\hat{u}_i = \sum_{j=0}^{r-1} \hat{u}_{j,i}^{(0)}\alpha^j,
\end{equation}
where $\alpha$ again is the primitive element of $\mathbb{F}_{2^r}$. For $1 \leq s \leq n$, the stage-$s$ binary estimations are computed by
\begin{equation}
	\begin{aligned}
		\hat{u}_{j,i}^{(s)} &= \hat{u}_{j,i}^{(s-1)} + \hat{u}_{j,i+2^{s-1}}^{(s-1)}, \\
		\hat{u}_{j,i+2^{s-1}}^{(s)} &= \hat{u}_{j,i+2^{s-1}}^{(s-1)}.
	\end{aligned}
\end{equation}
Therefore, the SC decoding of $\mathcal{C}_{\rm NB}$ is performed by running $r$ SC decoders and estimating $\hat{\bm{u}} = (\hat{u}_0, \hat{u}_1, \cdots, \hat{u}_{N-1})$ in a symbol-by-symbol manner. 

SCL decoding of $\mathcal{C}_{\rm NB}$ evolves from the above SC decoding by considering all possible values for information symbols, i.e.,
\begin{equation}
	\hat{u}_{i} =
	\begin{cases}
		\forall \sigma \in \mathbb{F}_{2^r}, \quad &\text{if}\, i\in \mathcal{A}; \\
		\sum_{t=0}^{\tau_i} \hat{u}^{\mathcal{A}}_t \cdot \bm{{\rm T}}_{t,i}, \quad &\text{if}\, i\in \mathcal{A}^c.
	\end{cases}
\end{equation}
For each information symbol, a decoding path is split into $2^r$ paths. In order to curb this exponentially increasing complexity, only the $L$ most likely paths are kept by the decoder. For this, a path metric should be defined to measure the likelihood of the paths. Assume that there are $L$ surviving paths after estimating $\hat{u}_{i-1}$. Their path metrics are defined as \cite{llr_scl}
\begin{equation} \label{froz}
	\Phi_{i-1}(l) = \sum_{h=0}^{i-1}\sum_{j:\hat{u}^{(0)}_{j,h}(l) \neq \eta_{j, h}(l)} |\mathcal{L}_{j, h}^{(0)}(l)|,
\end{equation}
where $0 \leq l \leq L-1$. Paths with smaller metrics are more likely to be correct. If $i\in \mathcal{A}$, the $l$-th decoding path is split into $2^r$ decoding paths with metrics
\begin{equation} \label{info}
	\Phi_i(l,\sigma) = \Phi_{i-1}(l) + \sum_{j:\sigma_j \neq \eta_{j, i}(l)} |\mathcal{L}_{j, i}^{(0)}(l)|, \quad \forall \sigma \in \mathbb{F}_{2^r}.
\end{equation}
Overall, $2^rL$ decoding paths are generated. A $(2^rL)$-to-$L$ path pruning is performed to select the $L$ paths with the smallest metrics. When $i = N-1$, the decoding path with the smallest metric will be selected as the decoding result. 

Algorithm \ref{scl} summarizes the above SCL decoding process. Complexity of the LLR computation and path sorting are $\mathcal{O}(rLN\text{log}_2N)$ and  $\mathcal{O}(K2^rL\text{log}_2(2^rL))$, respectively. Assume that operations of (\ref{llr_update}) are completed in one clock cycle and the $r$ SC decoders operate in parallel, the number of required clock cycles of the SCL decoding is $2N-2+KT_{\text{sort}}$, where $T_{\text{sort}}$ denotes the number of required clock cycles for path sorting.  

\begin{algorithm}[t] \label{scl}
	\caption{SCL Decoding of NB-LBCs}
	\For{$i = 0, 1, \cdots, N-1$}{
		Compute $\mathcal{L}_{j,i}^{(0)}(l)$ as in (\ref{llr_update}) \;
		\If{$i \in \mathcal{A}$}{
			Compute $\Phi_i(l,\sigma)$ as in (\ref{info})\;
			Select $L$ paths with the smallest path metrics\;
		}
		\Else{
			Compute $\hat{u}_{i}(l) = \sum_{t=0}^{\tau_i} \hat{u}^{\mathcal{A}}_t(l) \cdot \bm{{\rm T}}_{t,i}$\;
			Compute $\Phi_i(l)$ as in (\ref{froz})\;
		}
	}
	Select the decoding path with the smallest path metric\;
	Reconstruct $\hat{\bm{c}}$ as in (\ref{cug})\;
\end{algorithm}

\begin{algorithm}[t] \label{sort}
	\caption{Simplified Path Sorting}
	\KwIn{$\mathcal{L}_{j,i}^{(0)}(l)$,\,$\Phi_{i-1}(l)$;}
	\KwOut{$\Phi_{i}(l)$;}
	Sort $\Phi_{i-1}(0),\cdots,\Phi_{i-1}(L-1)$, yielding $\tilde{\bm{x}}$ and $\bm{p}$\;
	\For{$j = 0, 1, \cdots, r-1$}{
		$x_l^+ = \tilde{x}_l + |\mathcal{L}_{j,i}^{(0)}(p_l)|, l = 0, 1, \cdots, L-1$\;
		Sort $\bm{x}^+$, yielding $\tilde{\bm{x}}^+$ and $\bm{p}^+$\;
		Perform merge sort on $\tilde{\bm{x}}$ and $\tilde{\bm{x}}^+$ and keep the $L$ smallest values in $\tilde{\bm{x}}$\;
		Update $\bm{p}$ as the original indices of paths in $\tilde{\bm{x}}$\;
	}
	$\{\Phi_{i}(l) : l = 0, 1, \cdots, L-1\} \leftarrow \tilde{\bm{x}}$\;
\end{algorithm}

\subsection{Simplified Path Sorting} \label{reduce}
In the above SCL decoding, redundant comparisons exist in the full path sorting at each information symbol. E.g., let us consider $\beta, \gamma \in \mathbb{F}_{2^r}$ and their binary compositions only differ at the $j'$-th component, i.e., for $j = 0, 1, \cdots, j'-1, j'+1, \cdots, r-1$, $\beta_j = \gamma_j$. If $\beta_{j'} = \eta_{j', i}(l)$ and $\gamma_{j'} \neq \eta_{j', i}(l)$, then based on (\ref{info}), one can immediately obtain $\Phi_i(l,\beta) \leq \Phi_i(l,\gamma)$ without explicit comparison. This observation leads to a simplified path sorting \cite{fastSCL_journal}. It eliminates redundant path splits and unnecessary path metric comparisons during the full path sorting process.

It starts by sorting the path metrics of the surviving paths in ascending order. Let $\tilde{x}_l = \Phi_{i-1}(p_l)$, where $p_l \in \{0, 1, \cdots, L-1\}$ denotes the original path index corresponding to the $l$-th smallest path metric. The ordered metric vector and its associated index vector are given by
\begin{align}
	\tilde{\bm{x}} &= (\tilde{x}_0,\, \tilde{x}_1,\, \cdots,\, \tilde{x}_{L-1}), \\
	\bm{p} &= (p_0,\, p_1,\, \cdots,\, p_{L-1}),
\end{align}
where $\tilde{x}_0 \leq \tilde{x}_1 \leq \cdots \leq \tilde{x}_{L-1}$. 

At step-$j$, where $0 \leq j \leq r-1$, let us define $x_l^+ = \tilde{x}_l + |\mathcal{L}_{j,i}^{(0)}(p_l)|$, where $l = 0, 1, \cdots, L-1$. This yields the updated metric vector
\begin{equation}
	\bm{x}^+ = (x_0^+,\, x_1^+,\, \cdots,\, x_{L-1}^+).
\end{equation}
Sorting $\bm{x}^+$ in ascending order yields $\tilde{x}_l^+ = x_{p_l^+}^+$, where $p_l^+ \in \{0, 1, \cdots, L-1\}$ denotes the original index of the $l$-th smallest value in $\bm{x}^+$. The ordered version of $\bm{x}^+$ and its associated index vector are thus
\begin{align}
	\tilde{\bm{x}}^+ = (\tilde{x}_0^+,\, \tilde{x}_1^+,\, \cdots,\, \tilde{x}_{L-1}^+), \\
	\bm{p}^+ = (p_0^+,\, p_1^+,\, \cdots,\, p_{L-1}^+),
\end{align}
where $\tilde{x}_0^+ \leq \tilde{x}_1^+ \leq \cdots \leq \tilde{x}_{L-1}^+$. Since $\tilde{x}_l^+ = x_{p_l^+}^+$ and $x_{p_l^+}^+ =  \tilde{x}_{p_l^+} + |\mathcal{L}_{j,i}^{(0)}(p_{p_l^+})|$, the decoding path with metric $\tilde{x}_l^+$ is split from the ($p_{p_l^+}$)-th original path. By performing a merge sort on $\tilde{\bm{x}}$ and $\tilde{\bm{x}}^+$, the $L$ smallest values among the two sequences can be selected. The ordered metric vector $\tilde{\bm{x}}$ is then updated with these values and their original path indices are recorded in $\bm{p}$. The remaining elements in $\tilde{\bm{x}}$ and $\tilde{\bm{x}}^+$ are discarded, since path metrics increase monotonically at each step and any metric that is not among the $L$ smallest values at the current step cannot become one of the $L$ smallest in the subsequent steps.

After step-$(r-1)$, the $L$ paths with the smallest path metrics can be selected from all $2^rL$ paths. Algorithm \ref{sort} summarizes the above sorting process. During it, $(r+1)$ full sorts of length-$L$ sequences and $r$ merge sorts between two ordered length-$L$ sequences are required. The overall sorting complexity is reduced from $\mathcal{O}(2^rL\text{log}_2(2^rL))$ to $\mathcal{O}((r+1)L\text{log}_2L + rL)$.

\section{Simulation Results}

This section shows our simulation results on SCL decoding of eRS codes and NB-eBCH codes. They are obtained over the additive white Gaussian noise (AWGN) channel using BPSK modulation, where noise variance is $\frac{N_0}{2}$. The signal-to-noise ratio (SNR) is defined as $\frac{E_{\rm b}}{N_0}$, where $E_{\rm b}$ is the transmitted energy per information bit. Note that the Chase-BM decoding that filps the $\eta$ least reliable positions is denoted as $\text{Chase-BM}\,(\eta)$. The KV decoding with a maximum output list size of $L$ is denoted as $\text{KV}\,(L)$.

\begin{figure}[t]
	\centerline{\includegraphics[trim=0 15 0 20, width=0.50\textwidth]{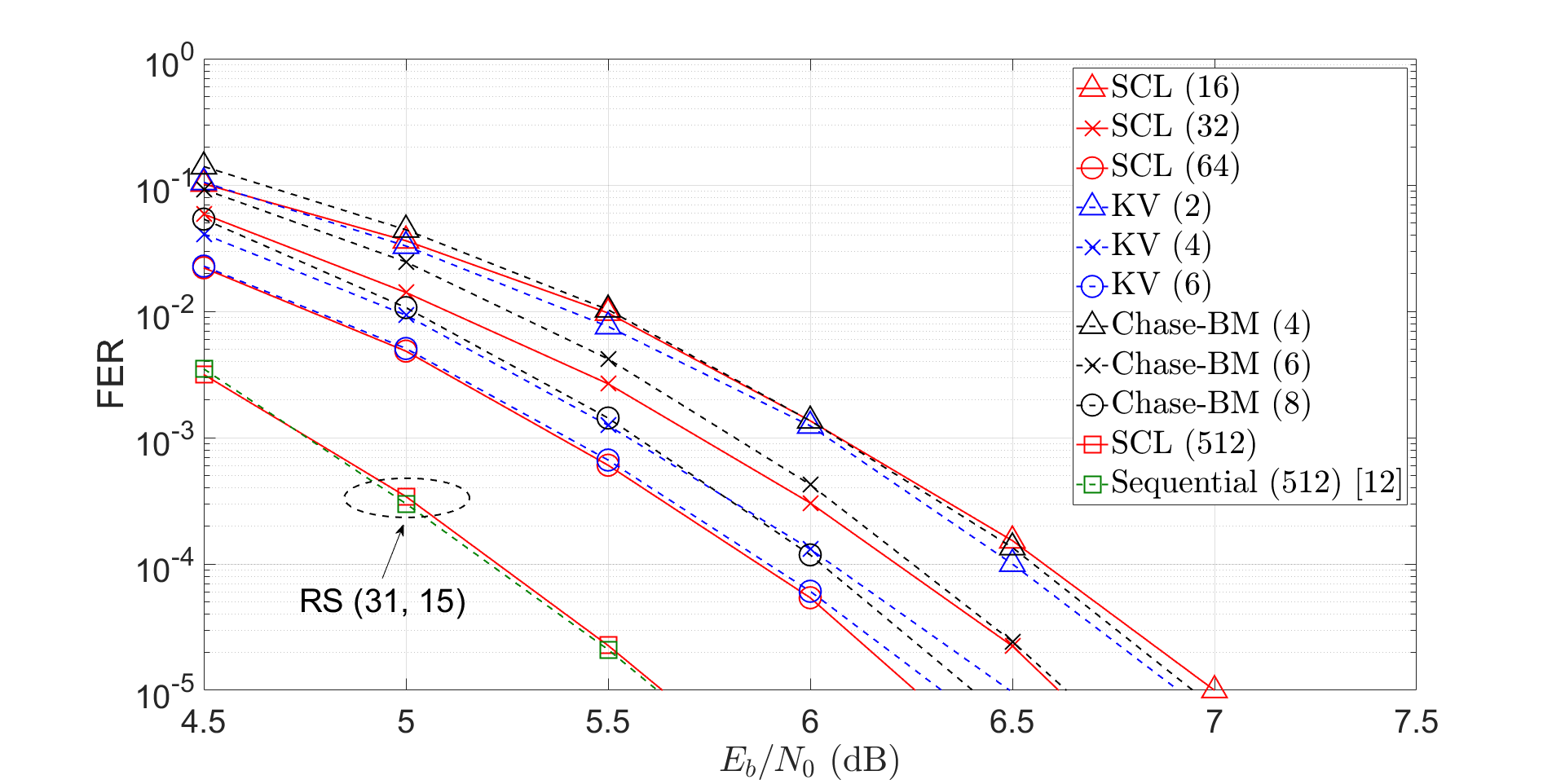}}
	\vspace{-3mm}
	\caption{SCL decoding performance of the $(32, 15)$ eRS code.}
	\label{eRS32}
\end{figure}

\begin{table}[t] 
	\caption{Decoding Complexity of the $(32, 15)$ eRS Code.}
	\vspace{-2mm}
	\centering
	\renewcommand{\arraystretch}{1.1}
	\label{complexity1}
	\footnotesize
	\begin{tabular}{ccc}\toprule[0.8pt]
		Scheme & $\mathbb{F}_{2^5}$ oper. & FLOPs  \\  \midrule
		$\text{SCL}(16)$    & $3.11\times10^{3}$ &  $1.41\times10^4$ \\ 
		$\text{SCL}(32)$    & $6.08\times10^{3}$ &  $3.27\times10^4$ \\ 
		$\text{SCL}(64)$    & $1.19\times10^{4}$ &  $7.32\times10^4$ \\ 
		$\text{KV}(2)$   & $6.37\times10^{4}$ & $6.80\times10^4$   \\ 
		$\text{KV}(4)$   & $5.06\times10^{5}$ & $1.13\times10^5$   \\ 
		$\text{KV}(6)$   & $2.05\times10^{6}$ & $1.59\times10^5$   \\ 
		$\text{Chase-BM}(4)$ & $4.93\times10^{4}$ & $7.99\times10^3$   \\ 
		$\text{Chase-BM}(6)$ & $2.00\times10^{5}$ & $8.71\times10^3$  \\ 
		$\text{Chase-BM}(8)$ & $8.08\times10^{5}$ & $1.14\times10^4$   \\ 
		\bottomrule[0.8pt]
	\end{tabular}
\end{table}

\begin{table}[!t] 
	\caption{FLOPs Required for SCL Decoding of the $(32, 15)$ eRS Code.}
	\vspace{-2mm}
	\centering
	\renewcommand{\arraystretch}{1.1}
	\label{reduction}
	\footnotesize
	\begin{tabular}{ccc}\toprule[0.8pt]
		& Full path sorting & Simplified path sorting   \\ \midrule
		$\text{SCL}(16)$   & $9.29\times10^{4}$ & $1.41\times10^4$ \\ 
		$\text{SCL}(32)$   & $2.05\times10^{5}$ & $3.27\times10^4$ \\ 
		$\text{SCL}(64)$   & $6.97\times10^{5}$ & $7.32\times10^4$ \\ 
		\bottomrule[0.8pt]
	\end{tabular}
\end{table}

Fig. \ref{eRS32} compares the frame error rate (FER) performance of the $(32, 15)$ eRS code under SCL, KV and Chase-BM  decoding. Table \ref{complexity1} further compares their decoding complexity at the SNR of $6\: \text{dB}$. It can be seen that SCL decoding can outperform KV and Chase-BM decoding with fewer finite field arithmetic operations. E.g., the $\text{SCL}\,(64)$  decoding not only slightly outperforms the $\text{KV}\,(6)$ decoding, but also reduces the number of finite field arithmetic operations by two orders of magnitude. Compared with the $\text{Chase-BM}\,(8)$ decoding, the $\text{SCL}\,(64)$ decoding yields a $0.1\,\text{dB}$ performance gain while reducing the number of finite field arithmetic operations by an order of magnitude. Fig. \ref{eRS32} also compares the performance of the $\text{SCL}\,(512)$ decoding and the sequential decoding with a maximum list size of $512$ \cite{RS_SC} on the $(31, 15)$ RS code, which can be seen as a punctured $(32, 15)$ eRS code. They achieve a similar performance. Sequential decoding with a maximum list size of $L$ has a worst-case complexity of $\mathcal{O}(L(N\text{log}_2^2N + N^2))$\cite{RS_SC}, which is higher than that of the $\text{SCL}\,(L)$ decoding, i.e., $\mathcal{O}(rLN\text{log}_2N)$.  Table \ref{reduction} shows that the proposed simplified path sorting reduces the number of FLOPs by an order of magnitude.

Fig. \ref{NBBCH} compares the FER performance of the $(64, 27)$ NB-eBCH code under the SCL, BM and Chase-BM decoding. The NB-eBCH code is defined over $\mathbb{F}_{4}$. Table \ref{complexity2} further compares their decoding complexity at the SNR of $6\: \text{dB}$. Similarly, SCL decoding reduces the number of finite field arithmetic operations at the cost of an increased number of FLOPs. E.g., while the $\text{SCL}\,(64)$ and $\text{Chase-BM}\,(10)$ decoding achieve similar FER performances, the $\text{SCL}\,(64)$ decoding requires two orders of magnitude fewer finite field arithmetic operations and an order of magnitude more FLOPs. 

\begin{figure}[t]
	\centerline{\includegraphics[trim=0 15 0 20, width=0.50\textwidth]{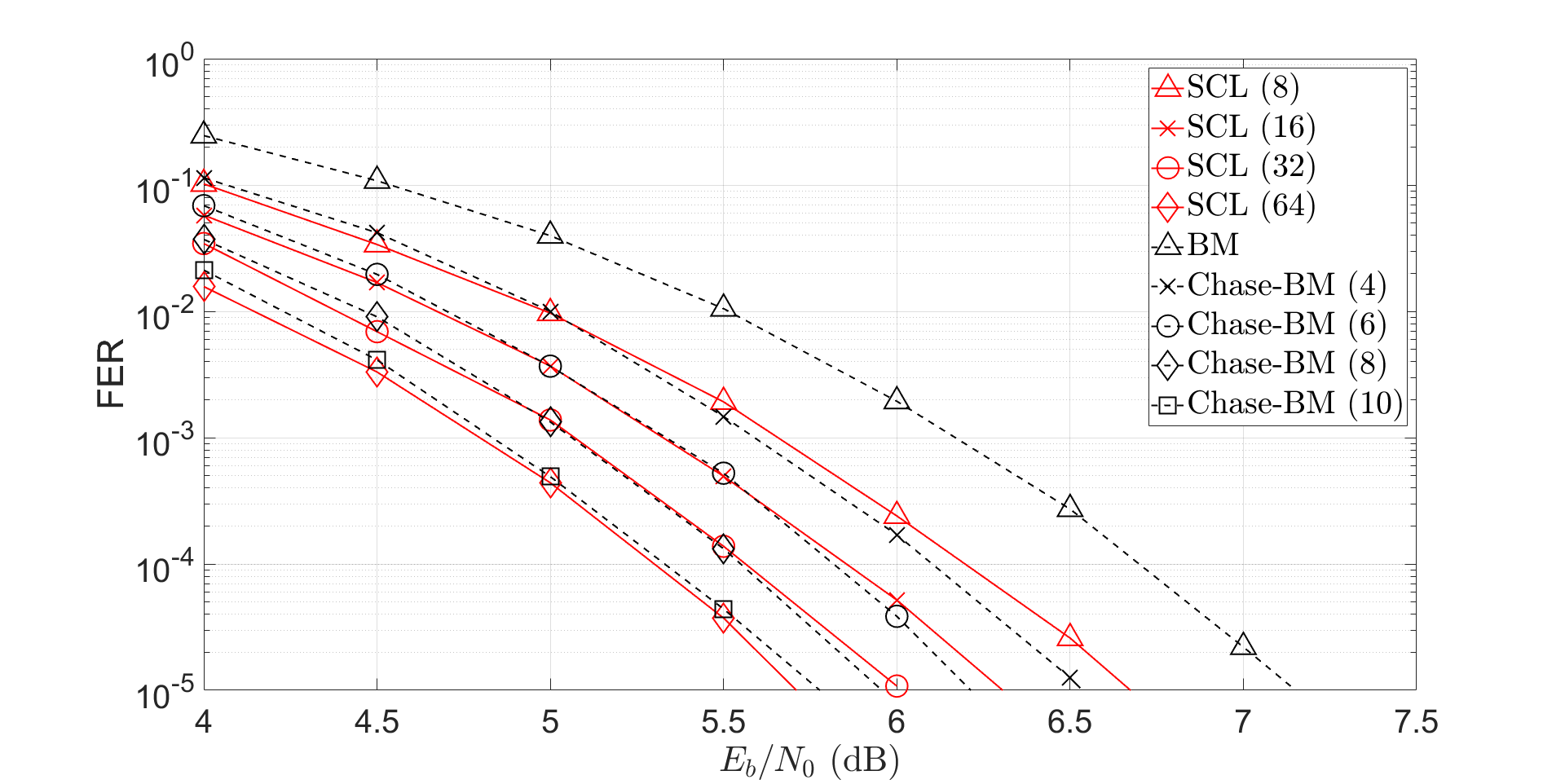}}
	\vspace{-3mm}
	\caption{SCL decoding performance of the $(64, 27)$ NB-eBCH code.}
	\label{NBBCH}
\end{figure}

\begin{table}[t] 
	\caption{Decoding Complexity of the $(64, 27)$ NB-eBCH Code.}
	\vspace{-2mm}
	\centering
	\renewcommand{\arraystretch}{1.1}
	\label{complexity2}
	\footnotesize
	\begin{tabular}{ccc}\toprule[0.8pt]
		Scheme & $\mathbb{F}_{4}$ oper. & FLOPs   \\ \midrule
		$\text{SCL}(16)$   & $8.31\times10^{3}$ & $1.52\times10^4$ \\ 
		$\text{SCL}(32)$   & $1.63\times10^{4}$ & $3.28\times10^4$ \\ 
		$\text{SCL}(64)$   & $3.23\times10^{4}$ & $7.13\times10^4$ \\ 
		$\text{Chase-BM}(6)$ & $4.63\times10^{5}$ & $2.77\times10^3$  \\ 
		$\text{Chase-BM}(8)$ & $1.87\times10^{6}$ & $4.04\times10^3$  \\ 
		$\text{Chase-BM}(10)$ & $7.57\times10^{6}$ & $9.06\times10^3$  \\ 
		\bottomrule[0.8pt]
	\end{tabular}
\end{table}

Finally, Fig. \ref{eRS16} shows the FER performance of SCL decoding of the $(16, 7)$ eRS code. The ML decoding upper and lower bounds \cite{RS_bound}, denoted as MLUB and MLLB, are also shown. When the list size $L=128$, the SCL decoding approaches the MLUB. However, as the code length increases, approaching the MLUB becomes increasingly difficult, since the SC decoding suffers from performance degradation when the permutation matrix of (\ref{permutation}) is applied. Therefore, the SCL decoding would require a larger list size to achieve a good performance.

\begin{figure}[t]
	\centerline{\includegraphics[trim=0 15 0 20, width=0.50\textwidth]{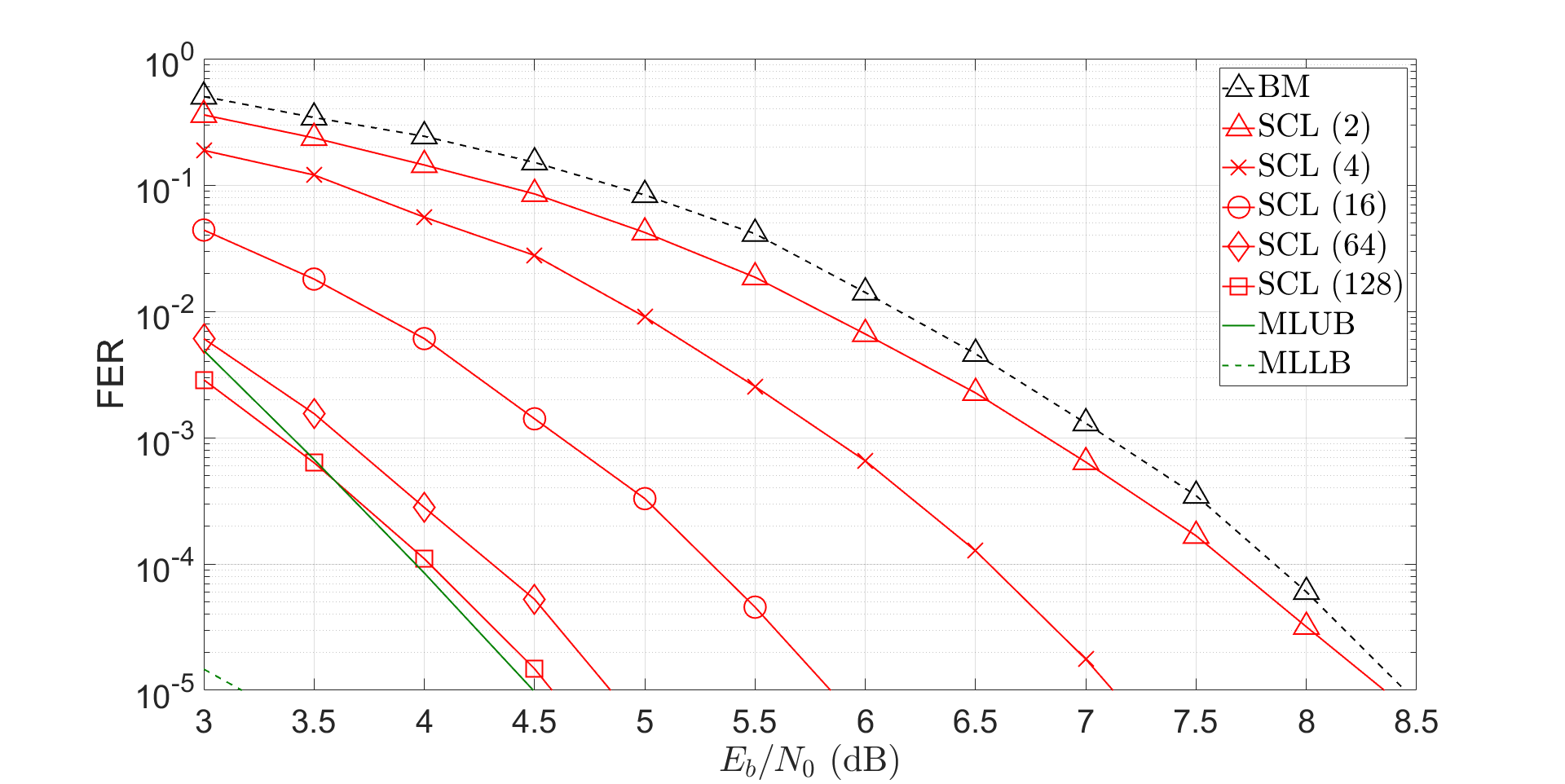}}
	\vspace{-3mm}
	\caption{SCL decoding performance of the $(16, 7)$ eRS code.}
	\label{eRS16}
\end{figure}

\bibliographystyle{IEEEtran} 
\bibliography{MyLibrary}

@article{bose1960class,
	author    = {Bose, Raj and Ray-Chaudhuri, Dwijendra},
	journal   = {Inf. and Contr.},
	title     = {On a class of error correcting binary group codes},
	year      = {1960},
	number    = {1},
	pages     = {68--79},
	volume    = {3},
	publisher = {Elsevier},
}

@article{hocquenghem1959codes,
	title={Codes correcteurs d'erreurs. Chiffres (Paris), 2, 147-156},
	author={Hocquenghem, A},
	journal={Math. Rev},
	volume={22},
	pages={652},
	year={1959}
}

@article{arikan2009polar,
	author  = {Arıkan, Erdal},
	journal = {IEEE Trans. Inf. Theory},
	title   = {Channel polarization: A method for constructing capacity-achieving codes for symmetric binary-input memoryless channels},
	year    = {2009},
	month = {Jun.},
	number  = {7},
	pages   = {3051-3073},
	volume  = {55},
}

@article{niu2012crc,
	author={K. Niu and K. Chen},
	journal={IEEE Commun. Lett.},
	title={{CRC}-aided decoding of polar codes},
	year={2012},
	month={Oct.},
	volume={16},
	number={10},
	pages={1668-1671},
}

@article{tal2015list,
	title={List decoding of polar codes},
	author={{Tal}, Ido and {Vardy}, Alexander},
	journal={{IEEE} Trans. Inf. Theory},
	volume={61},
	number={5},
	pages={2213--2226},
	year={2015},
	month={Mar.},
	publisher={IEEE}
}

@article{linhuang,
	author={Lin, Chien-Ying and Huang, Yu-Chih and Shieh, Shin-Lin and Chen, Po-Ning},
	journal={IEEE Open J. Commun. Soc.}, 
	title={Transformation of Binary Linear Block Codes to Polar Codes With Dynamic Frozen}, 
	year={2020},
	month={Mar.},
	volume={1},
	number={},
	pages={333-341},
	doi={10.1109/OJCOMS.2020.2979529}
}

@INPROCEEDINGS{dynamicfrozen,
	author={Trifonov, Peter and Miloslavskaya, Vera},
	booktitle={IEEE Inf. Theory Workshop (ITW)}, 
	title={Polar codes with dynamic frozen symbols and their decoding by directed search}, 
	year={2013},
	month={Sep.},
	address={Seville, Spain},
	volume={},
	number={},
	pages={1-5},
	doi={10.1109/ITW.2013.6691213}
}

@article{llr_scl,
	author={Balatsoukas-Stimming, Alexios and Parizi, Mani Bastani and Burg, Andreas},
	journal={IEEE Trans. Signal Process.}, 
	title={{LLR}-Based Successive Cancellation List Decoding of Polar Codes}, 
	year={2015},
	month={Jun.},
	volume={63},
	number={19},
	pages={5165-5179},
	doi={10.1109/TSP.2015.2439211}
}

@article{Chase,
	author={Chase, D.},
	journal={{IEEE} Trans. Inf. Theory}, 
	title={Class of algorithms for decoding block codes with channel measurement information}, 
	year={1972},
	month={Jan.},
	volume={18},
	number={1},
	pages={170-182},
	doi={10.1109/TIT.1972.1054746}}

@INPROCEEDINGS{RS_bound,
	author={M. EI-Khamy and R. J. McEliece},
	booktitle={42nd Allerton Conf. Commun. Control Comput.}, 
	title={Bounds on the average binary minimum distance and the maximum likelihood performance of {Reed-Solomon} codes}, 
	year={2004},
	month = {Sep.},
	address = {Monticello, U.S.A.},
	volume={},
	number={},
	pages={290-299},
}

@article{reed1960polynomial,
	title={Polynomial codes over certain finite fields},
	author={Reed, Irving S and Solomon, Gustave},
	journal={Journal of the society for industrial and applied mathematics},
	volume={8},
	number={2},
	pages={300--304},
	year={1960},
	publisher={SIAM}
}

@article{goppa1977codes,
	title={Codes associated with divisors},
	author={Goppa, Valerii Denisovich},
	journal={Problemy Peredachi Informatsii},
	volume={13},
	number={1},
	pages={33--39},
	year={1977},
	publisher={Russian Academy of Sciences, Branch of Informatics, Computer Equipment and~…}
}

@book{berlekamp,
	title={Algebraic coding theory},
	author={Berlekamp, Elwyn R},
	year={1968},
	publisher={New York: McGraw-Hill}
}

@article{massey2003shift,
	title={Shift-register synthesis and {BCH} decoding},
	author={Massey, James},
	journal={{IEEE} Trans. Inf. Theory},
	volume={15},
	number={1},
	pages={122--127},
	year={2003},
	month={Jan.},
	publisher={IEEE}
}

@article{GS,
	title={Improved decoding of {Reed-Solomon} and algebraic-geometry codes},
	author={Guruswami, Venkatesan and Sudan, Madhu},
	journal={{IEEE} Trans. Inf. Theory},
	volume={45},
	number={6},
	pages={1757-1767},
	month={Sep.},
	year={1999},
	publisher={IEEE}
}

@article{koetter2003algebraic,
	title={Algebraic soft-decision decoding of {Reed-Solomon} codes},
	author={Koetter, Ralf and Vardy, Alexander},
	journal={{IEEE} Trans. Inf. Theory},
	volume={49},
	number={11},
	pages={2809--2825},
	month={Nov.},
	year={2003},
	publisher={IEEE}
}

@ARTICLE{linhuang2,
	author={{Lin}, Chien-Ying and Huang , Yu-Chih and Shieh, Shin-Lin and Chen, Po-Ning },
	journal={{IEEE} Trans. Commun.}, 
	title={Toward Universal Decoding of Binary Linear Block Codes via Enhanced Polar Transformations}, 
	year={2025},
	month={Jul.},
	volume={73},
	number={11},
	pages={10117-10129},
	doi={10.1109/TCOMM.2025.3587044}
}

@ARTICLE{xingchen,
	author={Xing, Jiongyue and Chen, Li and Bossert, Martin},
	journal={IEEE Trans. Commun.}, 
	title={Progressive Algebraic Soft-Decision Decoding of {Reed–Solomon} Codes Using Module Minimization}, 
	year={2019},
	month={Nov},
	volume={67},
	number={11},
	pages={7379-7391},
	doi={10.1109/TCOMM.2019.2927207}}

@ARTICLE{RS_SC,
	author={Trifonov, Peter},
	journal={Probl. Inf. Transm. }, 
	title={Successive cancellation decoding of {Reed-Solomon} codes}, 
	year={2014},
	month={Oct.},
	volume={50},
	number={4},
	pages={303-312}
}

@INPROCEEDINGS{RS_SCP,
	author={{Trifonov}, Peter},
	booktitle={IEEE Inf. Theory Workshopp (ITW)}, 
	title={Successive cancellation permutation decoding of {Reed-Solomon} codes}, 
	year={2014},
	month={Nov.},
	address = {Hobart, Australia},
	volume={},
	number={},
	pages={386-390},
	doi={10.1109/ITW.2014.6970859}
}

@INPROCEEDINGS{RS_SCS,
	author={Miloslavskaya, Vera and Trifonov, Peter},
	booktitle={Int. Symp. Inf. Theory Appl. (ISITA)}, 
	title={Sequential decoding of {Reed-Solomon} codes}, 
	year={2014},
	month={Oct.},
	address = {Victoria, Canada},
	volume={},
	number={},
	pages={453-457},
	doi={}
}

@ARTICLE{broadcast,
	author={Sun, He and Viterbo, Emanuele and Dai, Bin and Liu, Rongke},
	journal={IEEE Trans. Broadcast.}, 
	title={Fast Decoding of Polar Codes for Digital Broadcasting Services in {5G}}, 
	year={2024},
	month={Jan.},
	volume={70},
	number={2},
	pages={731-738},
	doi={10.1109/TBC.2023.3345642}
}

@ARTICLE{Simplified,
	author={Alamdar-Yazdi, Amin and Kschischang, Frank R.},
	journal={IEEE Commun. Lett.}, 
	title={A Simplified Successive-Cancellation Decoder for Polar Codes}, 
	year={2011},
	month={Oct.},
	volume={15},
	number={12},
	pages={1378-1380},
	doi={10.1109/LCOMM.2011.101811.111480}
}

@ARTICLE{fastSCL_journal,
	author={Hashemi, Seyyed Ali and Condo, Carlo and Gross, Warren J.},
	journal={IEEE Trans. Signal Process.}, 
	title={Fast and Flexible Successive-Cancellation List Decoders for Polar Codes}, 
	year={2017},
	month={Aug.},
	volume={65},
	number={21},
	pages={5756-5769},
	doi={10.1109/TSP.2017.2740204}
}

\end{document}